\documentclass[a4paper,12pt,twoside]{article}
 
\usepackage[a4paper,left=2.5cm,right=2.5cm,top=2.5cm,bottom=3.0cm]{geometry}
\usepackage{amsmath}
\usepackage{amsthm}
\usepackage{amssymb}
\usepackage{graphicx}
\usepackage{graphics}
\usepackage{color}
\usepackage{caption}
\usepackage{subcaption}
\usepackage{comment}
\usepackage{authblk}
\usepackage{enumitem}
\usepackage[normalem]{ulem}

\newcommand{\univ}{\mathcal{U}}
\newcommand{\allsets}{\mathcal{S}}

\newcommand{\setcovi}{\mathcal{X}^{I}}

\newcommand{\path}[2]{PV_{#1,#2}}
\newcommand{\pathh}[1]{PH_{#1}}  
\newcommand{\source}[1]{s(#1)}

\newcommand{\comp}[1]{C_{#1}}

\newcommand{\seti}[1]{I_{#1}}
\newcommand{\setu}[1]{U_{#1}}
\newcommand{\setup}[1]{U'_{#1}}
\newcommand{\setv}[1]{V_{#1}}
\newcommand{\setvp}[1]{V'_{#1}}
\newcommand{\setz}{Z}

\newcommand{\veru}[2]{u_{#1,#2}}
\newcommand{\veruc}[1]{u_{#1}}
\newcommand{\verup}[2]{u'_{#1,#2}}
\newcommand{\verv}[2]{v_{#1,#2}}
\newcommand{\vervp}[2]{v'_{#1,#2}}

\newcommand{\st}{\hspace{0.1cm}\bigl|\bigr.\hspace{0.1cm}}
\newcommand{\cA}{\mathcal{A}}
\newcommand{\cV}{\mathcal{V}}
\newcommand{\cS}{\mathcal{S}}

\newcommand{\cW}{\mathcal{W}}
\newcommand{\cU}{\mathcal{U}}

\newcommand{\cF}{\mathcal{F}}
\newcommand{\cB}{\mathcal{B}}

\newcommand{\problemACT}{\textup{ACT}}
\newcommand{\problemCMR}{\textup{ST}}
\newcommand{\problemAllCMR}{\textup{All-ST}}
\newcommand{\problemCMRk}{\textup{STU}}
\newcommand{\problemCMRmin}{\textup{$\min$-ST}}
\newcommand{\problemCMRmax}{\textup{$\max$-ST}}
\newcommand{\problemTPE}{\textup{TPE}}
\newcommand{\problemSC}{\textup{SC}}
\newcommand{\card}[1]{\lvert #1\rvert}
\newcommand{\natplus}{\mathbb{N}^+}

\newcommand{\outN}[2]{N^{-}_{#1}(#2)}
\newcommand{\inN}[2]{N^{+}_{#1}(#2)}

\newtheorem{theorem}{Theorem}[section]
\newtheorem{lemma}{Lemma}[section]
\newtheorem{corollary}{Corollary}[section]

\DeclareCaptionLabelSeparator{dot}{\bf .\ }
\usepackage{helvet}

\begin{document}

\renewcommand{\thefootnote}{\fnsymbol{footnote}}

\title{\textsc{The Snow Team Problems
\\[2mm] {\large --- Clearing Directed Subgraphs by Mobile Agents ---\\~}
}}

\renewcommand\Affilfont{\small}

\author[1]{Dariusz Dereniowski}
\author[2]{Andrzej Lingas}
\author[1]{\authorcr Dorota Osula}
\author[3]{Mia Persson}
\author[4]{Pawe\l {} \.Zyli{\'n}ski}

\affil[1]{Faculty of Electronics, Telecommunications and Informatics\authorcr Gda\'nsk University of Technology, 80-233 Gda\'nsk, Poland}
\affil[2]{Department of Computer Science, Lund University, 221 00 Lund, Sweden}
\affil[3]{Department of Computer Science, Malm\"o University, 205 06 Malm\"o, Sweden}
\affil[4]{Institute of Informatics, University of Gda\'nsk, 80-309 Gda\'nsk, Poland}

\date{}
\maketitle
\thispagestyle{empty}

\begin{abstract}
We study several problems of clearing subgraphs by mobile agents in digraphs. The agents can move only along directed walks of a digraph and, depending on the variant, their initial positions may be pre-specified. In general, for a given subset~$\cS$ of vertices of a digraph $D$ and a positive integer $k$, the objective is to determine whether there is a subgraph $H=(\cV_H,\cA_H)$ of $D$ such that (a) $\cS \subseteq \cV_H$, (b) $H$ is the union of $k$ directed walks in $D$, and (c) the underlying graph of $H$ includes a Steiner tree for $\cS$ in $D$. We provide several results on the polynomial time tractability, hardness, and parameterized complexity of the problem.
\end{abstract}

\vspace*{0.2cm}

{\bf keywords:} graph searching, FPT-algorithm, NP-hardness, monomial

\section{Introduction}\label{sec:Introduction}

Consider a city, after a snowstorm, where all streets have been buried in snow completely, leaving a number of facilities disconnected. For snow teams, distributed within the city, the main battle is usually first to re-establish connectedness between these facilities. This motivates us to introduce a number of (theoretical) {\em snow team} problems in graphs. Herein, in the introduction section, we shall formalize only one of them, leaving the other variants to be stated and discussed subsequently.  

Let $D=(\cV,\cA,F,B)$ be a vertex-weighted digraph of order $n$ and size $m$, with two vertex-weight functions $F \colon \cV \rightarrow \{0,1\}$ and $B \colon \cV \rightarrow \mathbb{N}$, such that its underlying graph is connected. (Recall that the \emph{underlying graph} of $D$ is a simple graph with the same vertex set and its two vertices $u$ and $v$ being adjacent if and only if there is an arc between $u$ and $v$ in $D$.) The snow problem is modeled by $D$ as follows. The vertices of~$D$ correspond to street crossings while its arcs correspond to (one-way) streets, the set $\cF=F^{-1}(1)$ corresponds to locations of facilities, called also from now on {\em terminals}, and the set $\cB=B^{-1}(\natplus)$ corresponds to vertices, called from now on {\em snow team bases}, where a (positive) number of snow ploughs is placed (so we shall refer to the function $B$ as a~{\em plough-quantity} function). Let $\mathbf{k}_B=\sum_{v \in \cV}B(v)$ be the total number of snow ploughs placed in the digraph. 

\begin{description}
\item[The Snow Team problem] (\problemCMR)\\[2mm]  
Do there exist $\mathbf{k}_B$ directed walks in~$D$, with exactly $B(v)$ starting points at each vertex $v \in \cV$, whose edges induce a subgraph $H$ of $D$ such that all vertices in $F^{-1}(1)$ belong to one connected component of the underlying graph of $H$?
\end{description}

The \problemCMR\ problem may be understood as a question, whether for $\mathbf{k}_B$ snow ploughs, initially located at snow team bases in $\cB=B^{-1}(\natplus)$, where the number of snow ploughs located at $v\in \cB$ is equal to $B(v)$, it is possible to follow $\mathbf{k}_B$ walks in $D$ clearing their arcs so that  the underlying graph of the union of cleared walks includes a Steiner tree for all facilities in~$F^{-1}(1)$. 

\paragraph{Related work.} The Snow Team problem is related to the problems of clearing connections by mobile agents placed at some vertices in a digraph, introduced by Levcopoulos et al.\ in~\cite{ClearingConnectionsByFewAgents}. In particular, the \problemCMR\ problem is a generalized variant of the Agent Clearing Tree (ACT) problem where one wants to determine a placement of the minimum number of mobile agents in a digraph~$D$ such that agents, allowed to move only along directed walks, can simultaneously clear some subgraph of $D$ whose underlying graph includes a spanning tree of the underlying graph of~$D$. In~\cite{ClearingConnectionsByFewAgents}, the authors provided a~simple 2-approximation algorithm for solving the~Agent Clearing Tree problem, leaving its complexity status open.

All the aforementioned clearing problems are variants of the path cover problem in digraphs, where the objective is to find a minimum number of directed walks that cover all vertices (or edges) of a given digraph. Without any additional constraints, the problem was shown to be polynomially tractable by Ntafos and Hakimi in~\cite{NtafosHakimi79}. Several other variants involve additional constraints on walks as the part of the input, see~\cite{Beerenwinkel15, Gabow76,Kolman09,Kovac13,Ntafos1984, NtafosHakimi79,NtafosHakimi81} to mention just a few, some of them combined with relaxing the condition that all vertices of the digraph have to be covered by walks. In particular, we may be interested in covering only a given set of walks that themselves should appear as subwalks of some covering walks (polynomially tractable~\cite{NtafosHakimi79}). We may also be interested in covering only a given set of vertex pairs, where both elements of a pair should appear in the same order and in the same path in a solution (NP-complete even in acyclic digraphs~\cite{NtafosHakimi81}). Finally, for a given family $\cS$ of vertex subsets of $D$, we may be
interested in covering only a~representative from each of the subsets (NP-complete~\cite{NtafosHakimi81}). 

A wider perspective locates our snow team problems as variants of graph searching problems.
The first formulations by Parsons~\cite{Parsons76} and Petrov~\cite{Petrov82} of the first studied variant of these problems, namely the \emph{edge search}, were inspired by a work of Breisch~\cite{Breisch67}.
In~\cite{Breisch67}, the problem was presented as a search (conducted by a team of agents/rescuers) of a person lost in a system of caves.
The differences between the problem we study in this work and the edge search lies in the fact that in edge search the entity that needs to be found (usually called the \emph{fugitive}) changes its location quickly while our case can be interpreted as a search for an entity that is static. Also in the edge search, an agent can be removed from the graph and placed on any node (which is often referred to as \emph{jumping}) while in our problem it needs to follow a directed path.
A variant of the edge search that shares certain characteristics with the problem we study is the \emph{connected search}. In the latter, the connectivity restriction is expressed by requiring that at any time point, the subgraph that is ensured not to contain the fugitive is connected; for some recent algorithmic and structural results see e.g.~\cite{BFFFNST12,BestGTZ16,Dereniowski11,Dereniowski12}.
We also remark a different cleaning problem introduced in~\cite{MessingerNP08} and related to the variants we study: cleaning a graph with \emph{brushes} --- for some recent works, see e.g.~\cite{BorowieckiDP14,BryantFGPP14,GaspersMNP10,MessingerNP11}.
(Two restrictions from the original problem of cleaning a graph with brushes, namely, enforcing that each edge is traversed once and each cleaning entity must follow a walk in the graph appear also in a variant of edge search called the \emph{fast search}~\cite{DyerYY08}.) All aforementioned searching games are defined for simple graphs; for some works on digraphs, see e.g.~\cite{AlspachDHY07,AmiriKKRS15,BerwangerDHKO12,Oliveira16,HunterK08}.

Finally, the $\problemCMR$ problem is related to the directed Steiner tree problem, where for a~given edge-weighted directed graph $D=(\cV,\cA)$, a root $r\in\cV$ and a set of terminals $X\subseteq\cV$, the objective is to find a minimum cost arborescence rooted at $r$ and spanning all terminals in $X$ (equivalently, there exists a directed path from $r$ to each terminal in $X$) \cite{CharikarCCDGGL99,ZosinK02}.
For some recent works and results related to this problem, see e.g.~\cite{AbdiFGKS16,FriggstadKKLST14,JonesLRSS13}.
We also point out to a generalization of the Steiner tree problem in which pairs of terminals are given as an input and the goal is to find a minimum cost subgraph which provides a~connection for each pair \cite{CharikarCCDGGL99,FeldmannM16}.
For some other generalizations, see e.g.~\cite{ChitnisEHKKS14,Laekhanukit16,Suchy16,WatelWBB14,WatelWBB16}. 

\paragraph{Our results.} We show that the Snow Team problem as well as some of its variants are fixed-parameter tractable. In particular, we prove that the \problemCMR\ problem admits a~fixed-parameter algorithm with respect to the total number $l$ of facilities and snow team bases, running in $2^{O(l)} \cdot \textup{poly}(n)$ time, where $\textrm{poly}(n)$ is a polynomial in the order $n$ of the input graph  (Section~\ref{sec:theCMRproblemFPT}). The proof relies on the algebraic framework introduced by Koutis in~\cite{Koutis08}. On the other hand, we show that the \problemCMR\ problem (as well as some of its variants) is NP-complete, by a reduction from the Set Cover problem~\cite{GareyJohnson79} (Section~\ref{sec:theCMRproblemNP}).  Our result on NP-completeness of the \problemCMR\ problem implies NP-completeness of the Agent Clearing Tree problem studied in~\cite{ClearingConnectionsByFewAgents}, where the complexity status of the latter has been posed as an open problem.

\paragraph{Remark.} Note that a weaker version of the \problemCMR\ problem with the connectivity requirement removed, that is, when we require each facility only to be connected to some snow team base, admits a polynomial-time solution by a straightforward reduction to the minimum path cover problem in directed graphs~\cite{NtafosHakimi79}.

\paragraph{Notation.} The set of all source vertices in a directed graph $D$ is denoted by $\source{D}$. For a directed walk $\pi$ in $D$, the set of vertices (arcs) of $\pi$ is denoted by $V(\pi)$ (resp.\ $A(\pi)$), and its length --- by $|\pi|$. For two directed walks $\pi_1$ and $\pi_2$ in $D$, where $\pi_2$ starts at the ending point of $\pi_1$, the~concatenation of $\pi_1$ and $\pi_2$ is denoted by $\pi_1 \circ \pi_2$. 

Observe that in a border case, all non-zero length walks of snow ploughs start at the same vertex of the input digraph $D=(\cV,\cA,F,B)$. Therefore, we may assume that the number of snow ploughs at any vertex is at most $n-1$, that is, $B(v) \le n-1$ for any $v \in \cV$, and so the description of any input  requires $O(n\log n+m)$ space (recall $m \ge n-1$).

\newcommand{\ACfinal}[2]{Q(#1,#2)}
\newcommand{\AC}[3]{Q_{#1,#2}(#3)}
\newcommand{\ACin}[3]{Q^{+}_{#1,#2}(#3)}
\newcommand{\ACout}[3]{Q^{-}_{#1,#2}(#3)}
\newcommand{\indf}[2]{\mathbf{z}_{#2}(#1)}

\section{The \problemCMR\ problem is fixed-parameter tractable} \label{sec:theCMRproblemFPT}

In this section, we prove that the Snow Team problem is fixed-parameter tractable with respect to the number of facilities and snow team bases. The proof relies on the key fact (see Lemmas~\ref{lem:ST-AllST} and~\ref{lem:reduction-transitive} below) that a restricted variant of the \problemCMR\ problem with the~input~$D$ can be reduced to the detection of a particular directed subtree of `small' order in the transitive closure $\textup{TC}(D)$ of $D$.
We solve the latter tree detection problem by a~reduction to the problem of testing whether some properly defined multivariate polynomial has a monomial with specific properties, essentially modifying the construction in~\cite{KoutisW16} designed for undirected trees/graphs.  

Let us consider the variant of the \problemCMR\ problem, which we shall refer to as the \problemAllCMR\ problem, where we restrict the input only to digraphs $D=(\cV,\cA,F,B)$ that satisfy $\cB=B^{-1}(\natplus) \subseteq \cF=F^{-1}(1)$. (In other words, 
snow team bases can be located only at some facilities.)  
Observe that $D$ admits a positive answer to the \problemCMR\ problem if and only if there exists a subset $\cB'$ of $\cB \setminus \cF$ such that the digraph $D'=(\cV,\cA,F',B')$, where $F'(v)=1$ for $v \in \cF \cup \cB'$ and $F'(v)=0$ otherwise, and $B'(v)=B(v)$ for $v \in \cB' \cup (\cF \cap \cB)$ and $B'(v)=0$ otherwise, admits a positive answer to the \problemAllCMR\ problem. Therefore, we can immediately conclude with the following lemma.

\begin{lemma}\label{lem:ST-AllST}
Suppose that the \problemAllCMR\ problem can be solved in $2^{O(k)} \cdot \textup{poly}(n)$ time, where $k$ is the number of facilities in the input $($restricted$)$ digraph of order $n$. Then, the \problemCMR\ problem can be solved in $2^{O(l)} \cdot \textup{poly}(n)$ time, where $l$ is the total number of facilities and snow team bases in the input digraph of order $n$.\qed
\end{lemma}

Taking into account the above lemma, we now focus on constructing an efficient fixed-parameter algorithm for the \problemAllCMR\ problem, with the \emph{restricted} input digraph $D=(\cV,\cA,F,B)$ satisfying $\cB=B^{-1}(\natplus) \subseteq \cF=F^{-1}(1)$. Let $\cW$ be a set of walks (if any) that constitute a positive answer to the \problemAllCMR\ problem in $D$. We say that $\cW$ is {\em tree-like} if all walks in $\cW$ are strongly arc-distinct, that is, they are arc-distinct and if there is a walk in $\cW$ traversing the arc $(u,v) \in \cA$, then there is no walk in $\cW$ traversing its complement $(v,u)\in \cA$, and the underlying graph of their union is acyclic and includes a Steiner tree for $\cF$. Notice that if $\cW$ is tree-like, then all walks in $\cW$ are just (simple) paths.    

\begin{lemma}\label{lem:reduction-transitive}
A $($restricted$)$ instance $D=(\cV,\cA,F,B)$ admits a positive answer to the~\problemAllCMR\ problem if and only if the transitive closure $\textup{TC}(D)=(\cV,\cA',F,B)$ of $D$, with the same vertex-weight functions $F$ and $B$, admits a positive answer to the \problemAllCMR\ problem with a~tree-like set of walks whose underlying graph is of order at most $2|\cF|-1$.
\end{lemma}

\noindent Since the transitive closure $\textup{TC}(D)=(\cV,\cA',F,B)$ inherits the functions $F$ and $B$ from the restricted instance $D$, we emphasize that $\textup{TC}(D)$ is a proper (restricted) instance to the \problemAllCMR\ problem.    

\begin{proof} 
($\Leftarrow$) It follows from the fact that a directed walk in the transitive closure $\textup{TC}(D)$ corresponds to a directed walk in~$D$. 

\medskip\noindent
($\Rightarrow$) Assume that the snow ploughs initially located at vertices in $\cB$, with respect to the plough-quantity function $B$, can simultaneously follow $\mathbf{k}_B$ directed walks $\pi_1,\ldots,\pi_{\mathbf{k}_B}$ whose edges induce a subgraph $H$ of $D$ such that the underlying graph of $H$ includes a Steiner tree of $\cF$. Consider now the same walks in the transitive closure $\textup{TC}(D)$. To prove the existence of a tree-like solution of `small' order, the idea is to transform these $\mathbf{k}_B$ walks (if ever needed) into another strongly arc-disjoint $\mathbf{k}_B$ walks. The latter walks have the same starting points as the original ones (thus preserving the plough-quantity function~$B$) and the underlying graph of their union is a Steiner tree of $\cF$ (in the underlying graph of $\textup{TC}(D)$) having at most $|\cF|-1$ non-terminal vertices. 

\begin{figure}
\begin{center}

\centering
\includegraphics[scale=0.9]{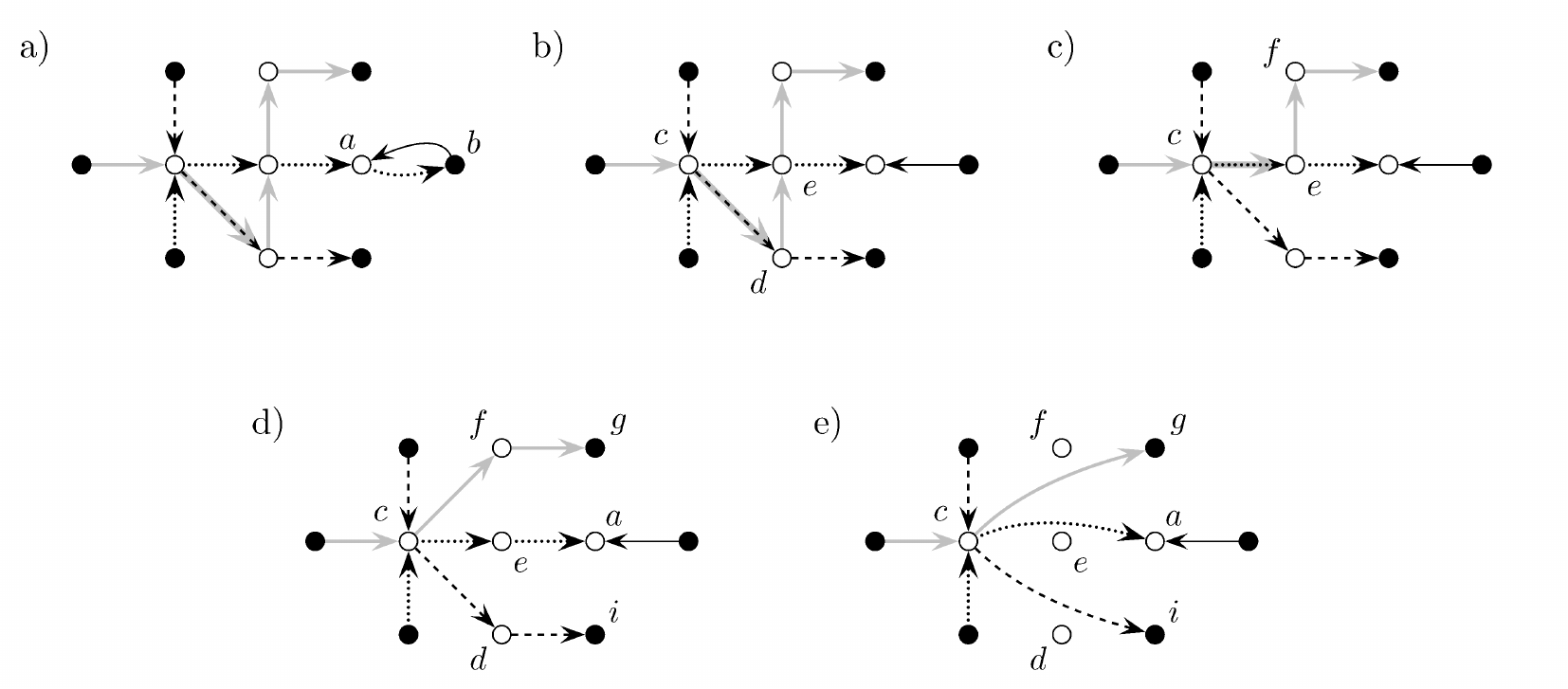}

\caption{Transforming walks into a tree-like set of walks. In the first step  (a-d), applied three times, we first delete the arc $(a,b)$, then replace two arcs $(c,d)$ and $(d,e)$ with the arc $(c,e)$, and then two arcs $(c,e)$ and $(e,f)$ with the arc $(c,f)$. In the second step (d-e), we replace two arcs $(c,d)$ and $(d,i)$,  two arcs $(c,e)$ and $(e,a)$, and two arcs $(c,f)$ and $(f,g)$, with the arcs $(c,i)$, $(c,a)$, and $(c,g)$, respectively.}\label{fig:2steps}
\end{center}
\end{figure}

Our transforming process is based on the following 2-step modification. First, assume without loss of generality that the walk $\pi_1=(v_1,\ldots, v_{|\pi_1|+1})$ has an arc $(v_t,v_{t+1})$  corresponding to an edge in the underlying graph $H$ of $\bigcup_{i=1}^{\mathbf{k}_B} \pi_i$ that belongs to a cycle (in $H$), or both $(v_t,v_{t+1})$ and its complement $(v_{t+1},v_t)$ are traversed by $\pi_1$, or there is another walk traversing $(v_t,v_{t+1})$ or $(v_{t+1},v_t)$, see Figure~\ref{fig:2steps} for an illustration. If $t=|\pi_1|$, then we shorten $\pi_1$ by deleting its last arc $(v_t,v_{t+1})$. Otherwise, if $t < |\pi_1|$, then we replace the arcs $(v_t,v_{t+1})$ and $(v_{t+1},v_{t+2})$ in $\pi_1$ with the arc $(v_{t},v_{t+2})$. One can observe that the underlying graph of the new set of walks is connected, includes a Steiner tree of $\cF$, and the vertex $v_1$ remains the starting vertex of (the new) $\pi_1$. But, making a walk cycle-free or strongly arc-disjoint may introduce another cycle in the underlying graph, or another multiply traversed arc,  or another arc $a$ such that both $a$ and its complement are traversed. However, the~length of the modified walk always decreases by one. Consequently, since the initial walks are of the finite lengths, we conclude that applying the above procedure multiple times eventually results in a tree-like set $\Pi=\{\pi_1,\ldots,\pi_{\mathbf{k}_B}\}$ of walks, being (simple) paths.

Assume now that in this set $\Pi$ of strongly arc disjoint paths, there is a non-terminal vertex $v$ of degree at most two in the underlying graph $H$ of $\bigcup_{i=1}^{\mathbf{k}_B}\pi_i$. Without loss of generality assume that $v$ belongs to the path $\pi_1$. Similarly as above, if $\deg_H(v)=1$, then we shorten $\pi_1$ by deleting its last arc. Otherwise, if $\deg_H(v) =2$ and $v$ is not the endpoint of $\pi_1$, then modify $\pi_1$ be replacing $v$ together with the two arcs of $\pi_1$ incident to it by the arc connecting the predecessor and successor of $v$ in $\pi_1$, respectively. Observe that since $v$ was a non-terminal vertex in the underlying graph, the underlying graph of (the new)  $\bigcup_{i=1}^{\mathbf{k}_B}\pi_i$ is another Steiner tree of $\cF$ (and does not include $v$). Moreover, the above modification keeps paths strongly arc-disjoint and does not change the starting vertex of $\pi_1$. Therefore, by subsequently replacing all such non-terminal vertices of degree at most two, we obtain a tree-like set of $\mathbf{k}_B$ paths in the transitive closure $\textup{TC}(D)$ such that the underlying graph of their union is a Steiner tree of $\cF$ with no degree two vertices except those either belonging to $\cF$ or being end-vertices of exactly two paths (in ${\rm TC}(D)$). Therefore, we conclude that the number of non-terminal vertices in this underlying graph is at most $|\cF|-1$, which completes our proof of the lemma. 

Indeed, the bound is obvious if $|\cF| \le 2$. So assume now that $|\cF| \ge 3$, the statement is valid for any set $\cF' \subset \cV$ with $0 \le |\cF'| < |\cF|$, and let $\Pi$ be a tree-like set of paths in the transitive closure ${\rm TC}(D)$ 
such that the underlying graph $H$ of their union is a Steiner tree $T$ of $\cF$ with no degree two vertices except those either belonging to $\cF$ or being end-vertices of exactly two paths. 
Let $v$ be a non-terminal vertex in $H$ (if no such $v$ exists, then there is nothing to prove), and let $\Pi' \subseteq \Pi$ be the set of paths ending at vertex $v$. By deleting all path arcs with endpoint $v$ for paths in $\Pi'$ and by replacing two consecutive path arcs incident to $v$ by the relevant arc connecting the predecessor and successor of $v$ in $\pi$, respectively, for any path $\pi \in \Pi \setminus \Pi'$, we obtain the set $\Pi'$ of strongly arc-distinct paths and a non-trivial partition $\cF_1 \cup \cdots \cup \cF_r=\cF$ such that the underlying graph of their union consists of $r\ge 2$ Steiner trees $T_1, ,\dots, T_r$ of $\cF_1,\ldots, \cF_r$, respectively, all of them with no degree two vertices except those either belonging to $\cF$ or being end-vertices of exactly two paths. By induction assumption, each tree $T_i$ has at most $|\cF_i|-1$ non-terminal vertices, $i=1, \ldots, r$, and so $T$ has at most $1+\sum_{i=1}^{r}(|\cF_i|-1) \le |\cF|-1$ non-terminal vertices since $r \ge 2$.       
\end{proof}

Taking into account the above lemma, a given (restricted) instance $D=(\cV,\cA,F,B)$ of the \problemAllCMR\ problem can be transformed (in polynomial time) into the answer-equivalent (restricted) instance $\textup{TC}(D)=(\cV,\cA', F, B)$ of the \emph{tree-like-restricted} variant of the \problemAllCMR\ problem in which only tree-like plough paths that together visit at most $2|\cF|-1$ vertices are allowed. Observe that $\textup{TC}(D)=(\cV,\cA', F, B)$ admits a positive answer to the tree-like-restricted \problemAllCMR\ problem if and only if $\textup{TC}(D)$ has a subtree $T=(\cV_T,\cA_T)$ of order at most $2|\cF|-1$ and such that $\cF \subseteq \cV_T$ and all edges of $T$ can be traversed by at most $\mathbf{k}_B$ snow ploughs following arc-distinct paths starting at vertices in $\cB$ (obeying the plough-quantity function~$B$). This motivates us to consider the following problem. 

Let $D=(\cV,\cA,F,B)$ be a directed graph of order $n$ and size $m$, with two vertex-weight functions $F \colon \cV \rightarrow \{0,1\}$ and $B \colon \cV \rightarrow \mathbb{N}$
such that $B^{-1}(\natplus) \subseteq F^{-1}(1)$, and let  $T=(V,A,L)$ be a directed vertex-weighted tree of order $\eta$, with a vertex-weight function $L \colon V \rightarrow \mathbb{N}$.

\begin{description}
\item[The Tree Pattern Embedding problem] (\problemTPE)\\[2mm]  
Does $D$ have a~subgraph $H=(\cV_H,\cA_H)$ isomorphic to $T$ such that $F^{-1}(1) \subseteq \cV_H$ and $L(v) \le B(h(v))$ for any vertex $v$ of $T$, where $h$ is an isomorphism of $T$ and $H$?
\end{description}

In Subsection~\ref{subsec:TPE}, we prove Theorem~\ref{thm:TPE-FPT} given below which states that there is a randomized algorithm that solves the \problemTPE\ problem in $O^\ast(2^{\eta})$ time, where  the notation $O^*$ suppresses polynomial terms in the order $n$ of the input graph $D$. We point out that if the order $\eta$ of $T$ is less than $|F^{-1}(1)|$ or at least $n+1$, then the problem becomes trivial, and so, in the following, we assume $|F^{-1}(1)| \le \eta \le n$. 

\begin{theorem}\label{thm:TPE-FPT}
There is a randomized algorithm that solves the~\problemTPE\ problem in $O^\ast(2^{\eta})$~time.
\end{theorem}

Suppose that for each vertex $v \in V$, the value $L(v)$ corresponds to the number of snow ploughs located at $v$ that are required to simultaneously traverse (clear) all arcs of $T$, in an arc-distinct manner, and $T$ admits a positive answer to the \problemTPE\ problem in the transitive closure $\textup{TC}(D)=(\cV,\cA',F,B)$. %of $D=(\cV,\cA,M,R)$, 
Then $\textup{TC}(D)$ admits a positive answer to the tree-like-restricted \problemAllCMR\ problem, which immediately implies that $D$ admits a positive answer to the \problemAllCMR\ problem (by Lemma~\ref{lem:reduction-transitive}). Therefore, taking into account Theorem~\ref{thm:TPE-FPT}, we are now ready to present the main theorem of this section. For simplicity of presentation, we now assume that a (restricted) directed graph $D=(\cV,\cA,F,B)$ itself (not its transitive closure) is an instance of the tree-like-restricted \problemAllCMR\ problem.

\begin{theorem} \label{thm:CMR-FPT}
There is a randomized algorithm that solves the tree-like-restricted \problemAllCMR\ problem for $D=(\cV,\cA,F,B)$ in $O^\ast(144^{|\cF|})$ time, where $\cF=F^{-1}(1)$.
\end{theorem}
\begin{proof}
Keeping in mind Lemma~\ref{lem:reduction-transitive}, we enumerate all undirected trees of order $\eta$, where $|\cF| \le  \eta \le 2|\cF|-1$ (and $\eta\le n$); there are $O(9^{|\cF|})$ such candidates~\cite{Otter48}. For each such a $\eta$-vertex candidate tree, we enumerate all orientations of its edges, in order to obtain a directed tree; there are $2^{\eta-1}$ such orientations. Therefore, we have $O(36^{|\cF|})$ candidates for a directed oriented tree $T$ of order $\eta$, where $|\cF| \le  \eta \le 2|\cF|-1$. 

For each candidate $T=(V,A)$, we determine in $O(\eta)$ time how many (at least) snow ploughs, together with their explicit locations at vertices in $V$, are needed to traverse all arcs of $T$, in an arc-disjoint manner. This problem can be solved in linear time just by noting that the number of snow ploughs needed at a vertex $v$ is equal to $\max\{0, \deg_{\textup{out}}(v)-\deg_{\textup{in}}(v)\}$ (since arcs must be traversed in an arc-disjoint manner). The locations of snow ploughs define a vertex-weight function $L \colon V \rightarrow \mathbb{N}$. We then solve the \problemTPE\ problem with the instance $D$ and $T=(V,A,L)$ in $O^\ast(2^{\eta})$ time by Theorem~\ref{thm:TPE-FPT}.   

As already observed, if $T$ admits a positive answer to the \problemTPE\ problem for $D$, then $D$ admits a positive answer to the tree-like-restricted \problemAllCMR\ problem. 
Therefore, by deciding the \problemTPE\ problem for each of $O(36^{|\cF|})$ candidates, taking into account the independence of any two tests, we obtain a randomized algorithm for the restricted \problemCMR\ problem with a running time $O^\ast(144^{|\cF|})$.
\end{proof}

Taking into account Lemma~\ref{lem:ST-AllST}, we immediately obtain the following corollary.

\begin{corollary} \label{cor:ST-FPT} 
The \problemCMR\ problem admits a fixed-parameter algorithm with respect to the total number $l$ of facilities and snow team bases, running in $2^{O(l)} \cdot \textup{poly}(n)$ time, where $n$ is the order of the input graph.\qed
\end{corollary}

\paragraph{Minimizing the number of used snow ploughs.} 
The first natural variation on the Snow Team problem is its minimization variant, which we shall refer to as the \problemCMRmin\ problem, where for a given input $n$-vertex digraph $D=(\cV,\cA,F,B)$, we wish to determine the minimum number of snow ploughs among those available at snow team bases in $\cB=B^{-1}(\natplus)$ that are enough to guarantee a positive answer the the (original) Snow Team problem in $D$. We claim that this problem also admits a fixed-parameter algorithm with respect to the total number $l$ of facilities and snow bases, running in time $2^{O(l)}\textup{poly}(n)$, and the solution is concealed in our  algorithm for the \problemCMR\ problem. Namely, observe that by enumerating all directed trees of order at most $|\cF|$, see the proof of Theorem~\ref{thm:CMR-FPT}, together with the relevant function $L$, and checking their embeddability in $D$, we accidentally solve this minimization problem: the embeddable tree with the minimum sum $\sum_{v \in V}L(v)$ constitutes the answer to \problemCMRmin\ problem. 

\begin{corollary}~\label{cor:STmin-FPT} 
The \problemCMRmin\ problem admits a randomized fixed-parameter algorithm with respect to the total number $l$ of facilities and snow team bases, running in $2^{O(l)} \cdot \textup{poly}(n)$ time, where $n$ is the order of the input graph.
\qed
\end{corollary}

\paragraph{Maximizing the number of re-connected facilities.} 
In the case when for the input digraph $D=(\cV,\cA,F,B)$, not all facilities can be re-connected into one component, that is, $D$ admits a negative answer to the Snow Team problem, one can ask about the maximum number of facilities in $F^{-1}(1)$ that can be re-connected by snow ploughs located with respect to the plough-quantity function $B$~\cite{MZ16}; we shall refer to this problem as the \problemCMRmax\ problem. Since we can enumerate all subsets of $\cF=F^{-1}(1)$ in $O^\ast(2^{|\cF|})$ time, taking into account Theorem~\ref{thm:CMR-FPT}, we obtain the following corollary. 

\begin{corollary}~\label{cor:STmax-FPT} 
The \problemCMRmax\ problem admits a randomized fixed-parameter algorithm with respect to the total number $l$ of facilities and snow team bases, running in $2^{O(l)} \cdot \textup{poly}(n)$ time, where $n$ is the order of the input graph.\qed
\end{corollary}

\paragraph{No pre-specified positions of snow ploughs.} 
Finally, another natural variant of the Snow Team problem is to allow any snow plough to start at any vertex. Formally, we define the following problem. 
\begin{description}
\item[The Snow Team problem with Unspecified snow team bases] (\problemCMRk)\\[2mm]  
Given a weight function $F \colon \cV \rightarrow \{0,1\}$ and an integer $k \ge 1$, do there exist $k$~directed walks in a digraph $D=(\cV,\cA)$ whose edges induce a subgraph $H$ of $D$ such that the set $F^{-1}(1)$ is a subset of the vertex set of $H$ and the underlying graph of $H$ is connected?
\end{description}

We claim that for the \problemCMRk\ problem, there is also  a randomized algorithm with the running time $2^{O(k+l)} \cdot \textup{poly}(n)$, where $l=|F^{-1}(1)|$ is the number of facilities, and  $n$ is the order of the input graph. The solution is analogous to that for the \problemCMR\ problem. Namely, one can prove a counterpart of Lemma~\ref{lem:reduction-transitive} which allows us to restrict ourselves to the restricted variant where only order $O(k+l)$ tree-like solutions are allowed. Then, the restricted variant is solved also using the algorithm for the \problemTPE\ problem as a subroutine: the function $B$ is the constant function $B(v)=n$, and among all directed tree candidates, we check only those with $\sum_{v \in V}L(v) \le k$. We omit details.

\begin{corollary}~\label{cor:CMRk-FPT} 
The \problemCMRk\ problem admits a randomized fixed-parameter algorithm with respect to the number $l$ of facilities and the number $k$ of snow ploughs, running in $2^{O(k+l)} \cdot \textup{poly}(n)$ time, where $n$ is the order of the input graph.\qed
\end{corollary}
 
Observe that if the number $k$ of available snow ploughs is not a part of the input, that is, we ask about the minimum number of walks whose underlying graph includes a Steiner tree for the set of facilities, then this problem seems not to be fixed-parameter tractable with respect to only the number of facilities. This follows from the fact that the minimum number of snow ploughs is unrelated to the number of facilities in the sense that even for two facilities to be connected, a lot of snow ploughs may be required, see Figure~\ref{fig:minST-notFPT} for an illustration.

\begin{figure}[!ht]
\begin{center}
\centering
\includegraphics[scale=0.9]{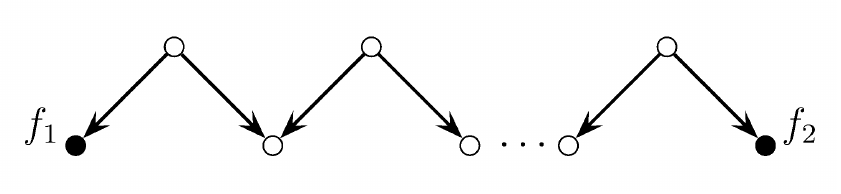}

\caption{Two facilities $f_1$ and $f_2$ require $n-1$ snow ploughs, where $n$ is the order of the digraph.}\label{fig:minST-notFPT}
\end{center}
\end{figure}

\subsection{The Tree Pattern Embedding problem} \label{subsec:TPE}

In this section, we solve the $\problemTPE$ problem by providing a randomized polynomial-time algorithm when the parameter $\eta$ is fixed. Our algorithm is based on the recent algebraic technique using the concepts of monotone arithmetic circuits and  monomials, 
introduced by Koutis in~\cite{Koutis08}, developed by Williams and Koutis in \cite{KoutisW16,W09}, and adapted to some other graph problems, e.g.,~\cite{BHKK17,BjorklundHT12,Kowalik16,FominLRSR12}.

A {\em $($monotone$)$ arithmetic circuit} is a directed acyclic graph where each leaf (i.e., vertex of in-degree $0$)
is labeled either with a variable or a real non-negative constant ({\em input gates}), each non-leaf vertex is  
labeled either with $+$ (an {\em addition gate} with an unbounded fan-in) or with $\times$ (a {\em multiplication gate} with fan-in two), and where a single vertex is distinguished (the {\em output gate}). Each vertex (gate) of the circuit represents (computes) a polynomial --- these are naturally defined by induction on the structure of the circuit starting from its input gates --- and we say that a polynomial is {\em represented} ({\em computed}) {\em by an arithmetic circuit} if it is represented (computed) by the output gate of the circuit. Finally, a polynomial that is just a product of variables is called a~{\em monomial}, and a~monomial in which each variable occurs at most once is termed a {\em multilinear monomial}~\cite{Koutis08,W09}.

We shall use a slight generalization of the main results of Koutis and Williams in~\cite{Koutis08,W09}, provided by them in Lemma~1 in~\cite{KoutisW16}, which, in terms of our notation, can be expressed as follows.

\begin{lemma}\label{lem:KW}{\em (\cite{KoutisW16})}
Let $P(x_1,\ldots,x_n,z)$ be a polynomial represented
by a monotone arithmetic circuit of size $s(n)$ and let $t$ be a non-negative integer. There is a randomized algorithm that for input $P$ runs in $O^*(2^kt^2s(n))$ time and outputs ``YES''
with high probability if there is a monomial of the form $z^{t} Q(x_1,\ldots,x_n)$,
where $Q(x_1,\ldots,x_n)$ is a multilinear
monomial of degree at most $k$ in the sum-product expansion of $P$, and
always outputs ``NO'' if there is no such  
monomial $z^{t} Q(x_1,\ldots,x_n)$ in the expansion.
\qed
\end{lemma}

Taking into account the above lemma, for the input digraph $D=(\cV,\cA,F,B)$ and directed tree $T=(V,A,L)$, the idea is to construct an appropriate polynomial $\ACfinal{X}{z}$ such that $\ACfinal{X}{z}$ contains a monomial of the form $z^{\card{\cS}}b(X)$, where $b(X)$ is a multilinear polynomial with exactly $|V|$ variables in $X$
and $\cS=F^{-1}(1) \cup B^{-1}(\natplus)$, if and only if the \problemTPE\ problem has a solution for the input $D$ and $T$ (see Lemma~\ref{lem:embedding} below).

\paragraph{Polynomial construction.} Let $D=(\cV,\cA,F,B)$ be a directed graph, with two vertex-weight functions $F \colon \cV \rightarrow \{0,1\}$ and $B \colon \cV \rightarrow \mathbb{N}$, and let  $T=(V,A,L)$ be a directed vertex-weighted tree of order $\eta$, with a vertex-weight function $L \colon V \rightarrow \mathbb{N}$. We consider $T$ to be rooted at a vertex $r \in V$, and for a non-root vertex $v$ of $T$, we denote the parent of $v$ in $T$ by $p(v)$. 
Now, for $v \in V$, define two sets $\inN{T}{v}$ and $\outN{T}{v}$:
\[
\begin{array}{l}
\inN{T}{v} = \left\{u\in V\st (u,v)\in A\textup{ and }u\neq p(v)\right\},\\[2mm]
\outN{T}{v} = \left\{u\in V\st (v,u)\in A\textup{ and }u \neq p(v)\right\}.
\end{array}
\]
%\[\inN{T}{v} = \left\{u\in V\st (u,v)\in A\textup{ and }u\neq p(v)\right\},
%\]
%\[\outN{T}{v} = \left\{u\in V\st (v,u)\in A\textup{ and }u \neq p(v)\right\}.
%\]

The idea is to treat $T$ as a `pattern' that we try to embed into the digraph $D$, with respect to functions $F,B$ and $L$. Denote $\cS=F^{-1}(1) \cup B^{-1}(\natplus)$ for brevity. We say that $T$ has an \emph{$\cS$-embedding into $D$} if the following holds (these are the formal conditions that need to be satisfied for the embedding to be correct):
\begin{enumerate}[label={\normalfont{(E\arabic*)}}]
\item\label{it:e1} There exists an injective function (homomorphism) $f\colon V\to\cV$ such that if $(u,v)\in A$, then $(f(u),f(v))\in\cA$.
\item\label{it:e2} $\cS \subseteq f(V)$, where $f(V)=\{f(v)\st v\in V\}$.
\item\label{it:e3} $L(v) \le B(f(v))$ for any $v \in V$.
\end{enumerate}

First, for $\cS\subseteq\cV$, $w\in\cV$ and $u \in V$, we introduce a particular indicator function, used for fulfilling Conditions~\ref{it:e2} and~\ref{it:e3}:
\[\indf{u,w}{\cS} =
\begin{cases}
z, & \textup{if }w\in \cS \textup{ and } L(u) \le B(w),\\
1, & \textup{if }w\notin \cS \textup{ and } L(u) \le B(w),\\
0, & \textup{otherwise, i.e., if }L(u) > B(w).
\end{cases}
\]
Next, following~\cite{KoutisW16}, we define a polynomial $\ACfinal{X}{T}$ that we then use to test existence of a desired $\cS$-embedding of $T$ in $D$.
Namely, we root $T$ at any vertex $r\in V$. Now, a polynomial $\AC{u}{w}{X}$, for a subtree $T_u$ of $T$ rooted at $u\in V$ and for a vertex $w\in\cV$, is defined inductively (in a bottom up fashion on $T$) as follows. For each $u\in V$ and for each $w\in \cV$: if $u$ is a leaf in $T$, then
\begin{equation} \label{eq:circuit1}
\AC{u}{w}{X} = \indf{u,w}{\cS}\cdot x_w,
\end{equation}
and if $u$ is not a leaf in $T$, then
\begin{equation} \label{eq:circuit2}
\AC{u}{w}{X} = 
\begin{cases}
\indf{u,w}{\cS}\cdot x_w\cdot\ACin{u}{w}{X} \cdot \ACout{u}{w}{X}, & \textup{if } \outN{T}{u}\neq\emptyset \land \inN{T}{u}\neq\emptyset, \\
\indf{u,w}{\cS}\cdot x_w\cdot\ACin{u}{w}{X},  & \textup{if } \outN{T}{u}=\emptyset, \\
\indf{u,w}{\cS}\cdot x_w\cdot\ACout{u}{w}{X},  & \textup{if } \inN{T}{u}=\emptyset,
\end{cases}
\end{equation}
where
\begin{equation} \label{eq:circuitIn}
\ACin{u}{w}{X} = \prod_{v\in\inN{T}{u}}\left(\sum_{(w',w)\in\cA}\AC{v}{w'}{X}\right),
\end{equation}
\begin{equation} \label{eq:circuitOut}
\ACout{u}{w}{X} = \prod_{v\in\outN{T}{u}}\left(\sum_{(w,w')\in\cA}\AC{v}{w'}{X}\right).
\end{equation}
Finally, the polynomial $\ACfinal{X}{z}$ is as follows:
\begin{equation}\label{eq:circuitFinal}
\ACfinal{X}{z} = \sum_{w\in\cV} \AC{r}{w}{X}.
\end{equation}

We have the following lemma.
\begin{lemma} \label{lem:embedding}
The polynomial $\ACfinal{X}{z}$ contains a monomial of the form $z^{\card{\cS}}b(X)$, where $b(X)$ is a multilinear polynomial with exactly $\eta$ variables in $X$, if and only if the $\eta$-vertex tree $T$ has an $\cS$-embedding into $D$.
\end{lemma}
\begin{proof}
Consider a vertex $u$ of $T$ and assume that the subtree $T_u$ is of order $j$.
Observe that, by a straightforward induction on the size of a subtree, a monomial $z^q x_{w_1}\cdots x_{w_{j}}$, where $w_i\in\cV$ for each $i\in\{1,\ldots,j\}$, is present in $\AC{u}{w_1}{X}$ if and only if the three following conditions hold.
\begin{enumerate}[label={\normalfont{(\roman*)}}]
\item\label{it:ind1} There exists a homomorphism $f_u$ from the vertices of $T_u$ to $w_1,\ldots,w_j$ such that $f_u(u)=w_1$.
\item\label{it:ind2} $\card{\cS\cap\{w_1,\ldots,w_j\}}\le q$ and the equality holds if $w_1,\ldots,w_j$ are pairwise different.
\item\label{it:ind3} $L(v) \le B(f_u(v))$ for any vertex $v$ of $T_u$.
\end{enumerate}
The fact that $f_u$ is a homomorphism follows from the observation that, during construction of $\AC{u}{w_1}{X}$ in~\eqref{eq:circuitIn} and~\eqref{eq:circuitOut}, a neighbor $v$ of $u$ is mapped to a node $w'$ of $D$ in such a way that if $(v,u)\in A$ then $(w',w)\in\cA$ (see~\eqref{eq:circuitIn}), and if $(u,v)\in A$ then $(w,w')\in\cA$ (see~\eqref{eq:circuitOut}).
Conditions~\ref{it:ind2} and~\ref{it:ind3} are ensured by appropriate usage of the indicator function in~\eqref{eq:circuit1}, namely, if $u$ is mapped to $w$ in a homomorphism corresponding to $\AC{u}{w}{X}$, then we add the multiplicative factor of $z$ to $\AC{u}{w}{X}$ provided that $L(v) \le B(w)$.

Thus, we obtain that $\ACfinal{X}{z}$ has a multilinear polynomial $z^{\card{\cS}}x_{w_1}\cdots x_{w_{\eta}}$ if and only if $T$ has an $\cS$-embedding into $D$.
\end{proof}

Observe that the polynomial $\ACfinal{X}{z}$ and the auxiliary polynomials $\ACin{u}{w}{X},$ $\ACout{u}{w}{X}$ can be represented by a monotone arithmetic circuit of size polynomial in the order $n$ of the input digraph $D$. To start with, we need $n+1$ input gates for the variables corresponding to vertices of $D$, and the auxiliary variable $z$. With each of the aforementioned polynomials, we associate a gate representing it, which gives in total $O(\eta n)$ such gates. In order to implement the recurrences defining the polynomials, assuming unbounded fan-in of addition gates, we need $O(n)$ auxiliary gates for each recurrence involving large products. Thus, the resulting circuit is of size $O(n^3)$. Hence, by Lemma~\ref{lem:KW} combined with Lemma \ref{lem:embedding}, we conclude that the existence of an $\cS$-embedding of the $\eta$-vertex tree $T$ into $D$ can be decided in $O^*(2^{|\cS|})$ time. Consequently, since $|\cS|\le \eta$, we obtain Theorem \ref{thm:TPE-FPT} by the definition of an $\cS$-embedding.

\paragraph{Remark.} The above approach can be adapted to the case when we want to embed a directed forest $T=(\cV,\cA,F,B)$ of order $\eta$ into a directed graph. All we need is to build a relevant polynomial for each rooted directed tree-component of $T$, and then to consider the product $S(X,T)$ of these polynomials, asking about the existence of a monomial of the form $z^{\card{\cS}}b(X)$, where $b(X)$ is a multilinear polynomial with exactly $\eta$ variables in~$X$. Also, by a similar approach, we may consider and can solve (simpler) variants of our embedding problem without the weight function $F$ or without the weight functions $B$ and~$L$; details are omitted.

\section{The Snow Team problem is hard} \label{sec:theCMRproblemNP}

In this section, we prove that the Snow Team problem is NP-complete by describing a~polynomial-time reduction from the Set Cover problem. 

Let  $\univ=\{u_1,\ldots,u_n\}$ be a set of $n$ items and let $\allsets=\{S_1,\dots,S_m\}$ be a family of $m$ sets containing the items in $\cU$, i.e., each $S_t\subseteq\univ$, such that each element in $\univ$ belongs to at least one set from $\allsets$. A~$k$-element subset of $\allsets$, whose union is equal to the whole universe $\univ$, is called \textit{a set cover of size $k$.}

\begin{description}
\item[The Set Cover problem] (\problemSC)\\[2mm]  
Given $\univ,\allsets$ and a positive integer $k$, does there exist a set cover of a size $k$?
\end{description}

The Set Cover problem is well known to be NP-complete~\cite{GareyJohnson79}. We are going to prove that for a given $\univ=\{u_1,\ldots,u_n\},\allsets=\{S_1,\dots,S_m\}$ and an integer $k$, there exists a set cover of size $k$ if and only if there is a solution for the $\problemCMR$ problem in the appropriately constructed acyclic digraph $D_{\textup{SC}}=(\cV,\cA,F,B)$. Basically, in this graph, for each element $u \in \univ$, there is a gadget $C_u$ being the union of the number of `vertical' paths equal to the number of sets that $u$ appears in. Also, there is one `spanning' gadget including $m$ `horizontal' paths $\pathh{t}$, $t=1,\ldots,m$, each of which visits the relevant gadget $C_u$ if the element $u$ belongs to the set $S_t$. In our construction, all vertices are facilities, and snow teams are located only at source vertices, one team at each source vertex (see details below).

In the following, we assume $\univ=\{1, \ldots, n\}$ and that elements in $S_t=\{x_1,\ldots,x_{\ell(t)}\}$ are sorted in the ascending order, where $\ell(t)$ denotes the size of $S_t$, $t \in \{1,\ldots, m\}$. Also, we denote a solution to the set cover problem with the input $\langle\univ,\allsets,k\rangle$ by $\setcovi$ which is encoded as a subset of~$\{1,\ldots,m\}$, where $t\in\setcovi$ if and only if $S_t$ belongs to the $k$-element subset of $\allsets$ forming the set cover. Finally, for simplicity of presentation, we denote a set $\{1,\ldots, l\}$ of indices by $[l]$.

\paragraph{Digraph construction.} For $i \in \cU$, let $\seti{i}$ be the set of all indices of subsets from~$\allsets$ which contain $i$: $I_i=\{j\st i\in S_j\}$. First, for every $i \in \cU$, we introduce the following four sets of vertices (see Figure~\ref{fig:nphard1} for an illustration):
\begin{itemize}
\item $\setu{i} = \left\{ \veruc{i} \right\} \cup \left\{\ \veru{i}{j} :  j \in \seti{i} \right\},$
\item $\setup{i} = \left\{\ \verup{i}{j} :  j \in \seti{i} \right\},$
\item $\setv{i} = \left\{\ \verv{i}{j} :  j \in \seti{i} \right\},$
\item $\setvp{i} = \left\{\ \vervp{i}{j} :  j \in \seti{i} \right\}.$
\end{itemize}
Next, we introduce an additional set $\setz = \{z, z_1, \ldots, z_k\}$ of vertices corresponding to the input integer~$k$. Finally, for every $i \in \cU$ and $j \in \seti{i}$, we create a directed path $\path{i}{j} = (\veru{i}{j}, \veru{i}{}, \verup{i}{j}, \verv{i}{j}, \vervp{i}{j})$; we refer to these paths as \emph{vertical}. Observe that for $i \in \cU$, the union $\bigcup_{j \in \seti{i}} V(\path{i}{j})$ induces the directed subgraph which we shall refer to as the \textit{i-th element component} $\comp{i}$.
This finishes the first step of our construction. 

\begin{figure}[!htb]
\begin{center}
  \centering
  \includegraphics[scale=0.75]{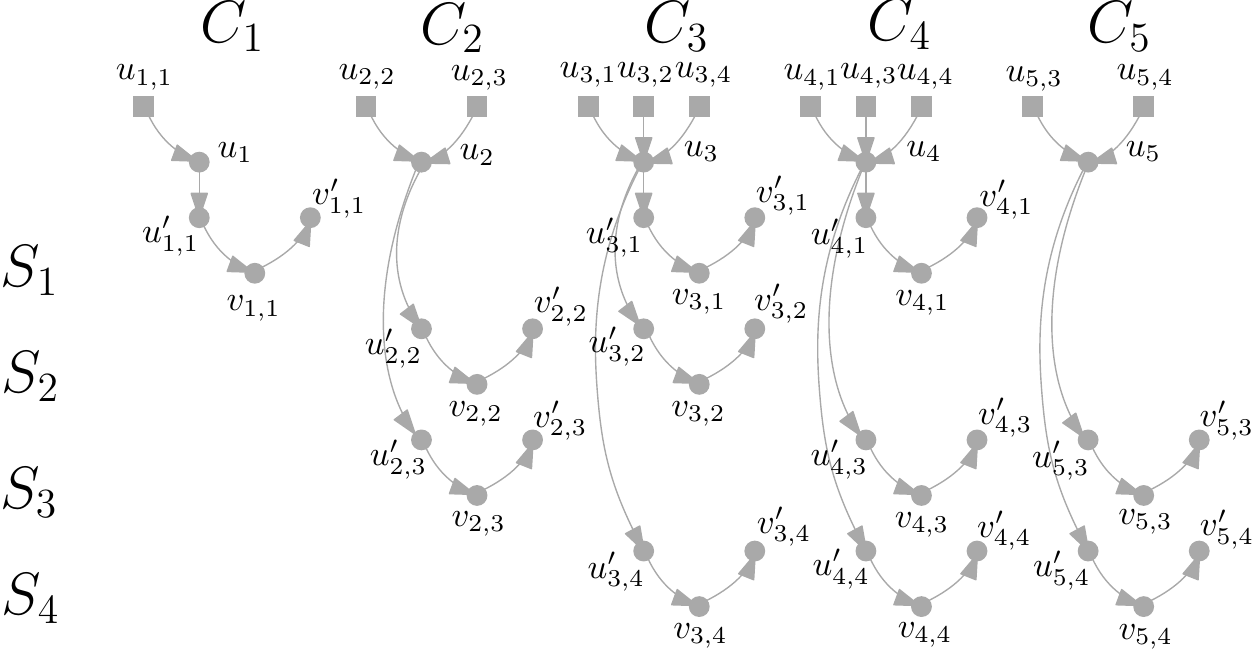}
\caption{Step 1: construction of the element components. Here, $\univ = \{1, 2, 3, 4, 5\}$ and $ \allsets=\{S_1,S_2,S_3,S_4\}$, where $S_1 = \{1, 3, 4\}$, $S_2 = \{2, 3\}$, $\ S_3=\{2, 4, 5\}$ and $S_4 = \{3, 4, 5\}$. The corresponding sets of indices are: $\seti{1}=\{ 1\},\ \seti{2}=\{2,3\},\ \seti{3}=\{1,2,4\},\ \seti{4}=\{1,3,4\}$ and $\seti{5}=\{3,4\}$. For $k=2$, there exists a set cover $\setcovi = \{1, 3\}$ of size $2$, consisting of the sets $S_1$ and $S_3$.}
\label{fig:nphard1}
\end{center}
\end{figure}

For the second step (see Figure~\ref{fig:nphard2} for an illustration), for $t \in [m]$, we consider the set $S_t=\{x_1,\ldots,x_{\ell(t)}\}$ --- recall that the elements in $S_t$ are sorted in the ascending order --- and construct a directed path
$\pathh{t} = (z,\verv{x_{1}}{t},\verv{x_{2}}{t},\ldots,\verv{x_{\ell(t)}}{t})$; these paths are called \emph{horizontal}. Next, we add $k$ additional arcs, namely, for each $l\in[k]$, we add the arc $(z_l,z)$. 

With our 2-step construction, we have built the directed graph $D_{\textup{SC}} = \left( \mathcal{V}, \mathcal{A} \right)$, where
\[
\cV = \bigcup\limits_{i \in [n]} V(\comp{i}) \cup \setz \quad \textrm{and} \quad \cA = \left\{(z_i,z)\st i\in[k]\right\}\cup\bigcup\limits_{i \in [n]} E(\comp{i}) \cup \bigcup\limits_{i \in [m]} E(\pathh{i}).
\]
We finalize our construction by defining the functions $F  \colon \cV \rightarrow \{0,1\}$ and $B \colon \cV \rightarrow \mathbb{N}$. Specifically, we set $F(v)=1$ for each $v \in \cV$ , and $B(v)=1$ if and only if $v \in \cV$ is a~source vertex in $D_{\textup{SC}}$:   
\[B(v) = \begin{cases}
1, & \textup{if }v\in \left\{\veru{i}{j}\st i\in\cU,j\in I_i\right\}\cup\{z_1,\ldots,z_k\} \\
0, & \textup{otherwise}.
\end{cases}
\]
Therefore, $\cF=\mathcal{V}$ and $\cB=B^{-1}(\natplus)$ is the set of all source vertices in $D_{\textup{SC}}$. In  particular, there is exactly one snow plough at each source vertex of $D_{\textup{SC}}$, and so $\mathbf{k}_B = \sum_{v \in \cV}B(v)= |\source{D_{\textup{SC}}}|$, which equals to $k + \sum_{i=1}^m |S_i|$ by the construction of~$D_{\textup{SC}}$. 

Clearly, our reduction takes polynomial time. The order of $D_{\textup{SC}}$ is equal to $4\cdot \sum_{i=1}^m|S_i|+n+k+1=O(nm+k)$, its size also equals $O(nm+k)$, and the descriptions of the functions $F$ and $B$ require $O(nm+k)$ space either.
%$O(n)$ space. 
Finally, observe that $D_{\textup{SC}}$ is acyclic and its underlying graph is connected. The latter observation follows from the fact that $\allsets$ is a~family of sets whose union is $\univ$, and each element in $\univ$ belongs to at least one set from~$\allsets$.

\begin{figure}[!tb]
\begin{center}
  \centering
  \includegraphics[scale=0.75]{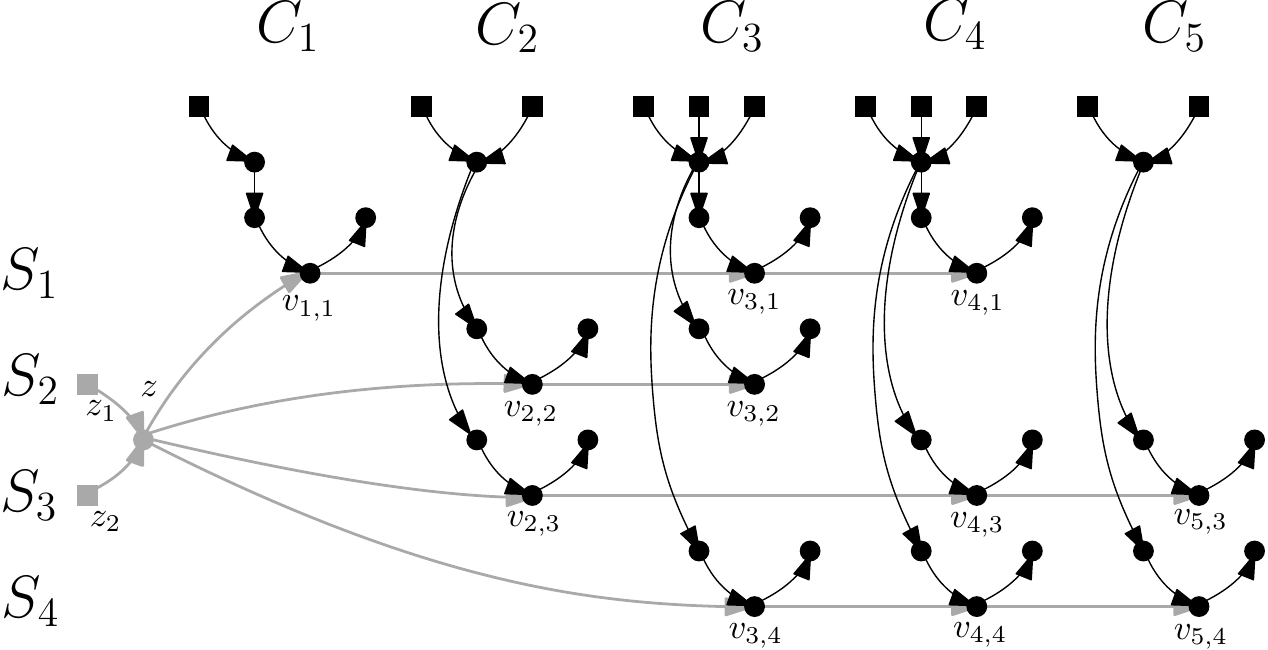}
  \caption{(Cont.\ Figure~\ref{fig:nphard1}) Step 2: four horizontal paths connecting all components and two new source vertices $z_1, z_2$ are added (recall $k=2$).}
\label{fig:nphard2}
\end{center}
\end{figure}

\subsection{Direct implication} %--------------------------------------------------------------------------------

First, we are going to prove the direct implication.

\begin{lemma}\label{lem:from_sc_to_act}
Let $\langle \univ,\allsets,k\rangle$ be an instance of the \problemSC\ problem. If there exists a set cover of size $k$ for $\univ$ and $\allsets$, then there exists a solution to the \problemCMR\ problem for the digraph $D_{\textup{SC}}=(\cV,\cA,F,B)$.
\end{lemma}

\begin{proof}
(See Figure~\ref{fig:nphard3} for an illustration.) Let $\setcovi=\{\xi(1),\ldots,\xi(k)\}$ be a solution to the \problemSC\ problem for $\langle\univ,\allsets,k\rangle$. We now construct a solution $\cW$ to the \problemCMR\ problem to consist of the following paths:
\[\path{i}{j}=(\veru{i}{j}, \veru{i}{}, \verup{i}{j}, \verv{i}{j}, \vervp{i}{j}), \quad \textrm{ for } i\in\univ,\ j\in\seti{i},\]
and
\[(z_t,z) \circ \pathh{\xi(t)}=(z_t,z,\verv{x_{1}}{\xi(t)},\verv{x_{2}}{\xi(t)},\ldots,\verv{x_{\ell(t)}}{\xi(t)}), \quad \textrm{ for } t\in [k],\]
where $x_1,\ldots,x_{\ell(t)}$ are the (ordered) elements of the set $S_t \in \cS$, with $|S_t|=\ell(t)$. 

\begin{figure}[!tb]
\begin{center}
  \centering
  \includegraphics[scale=0.75]{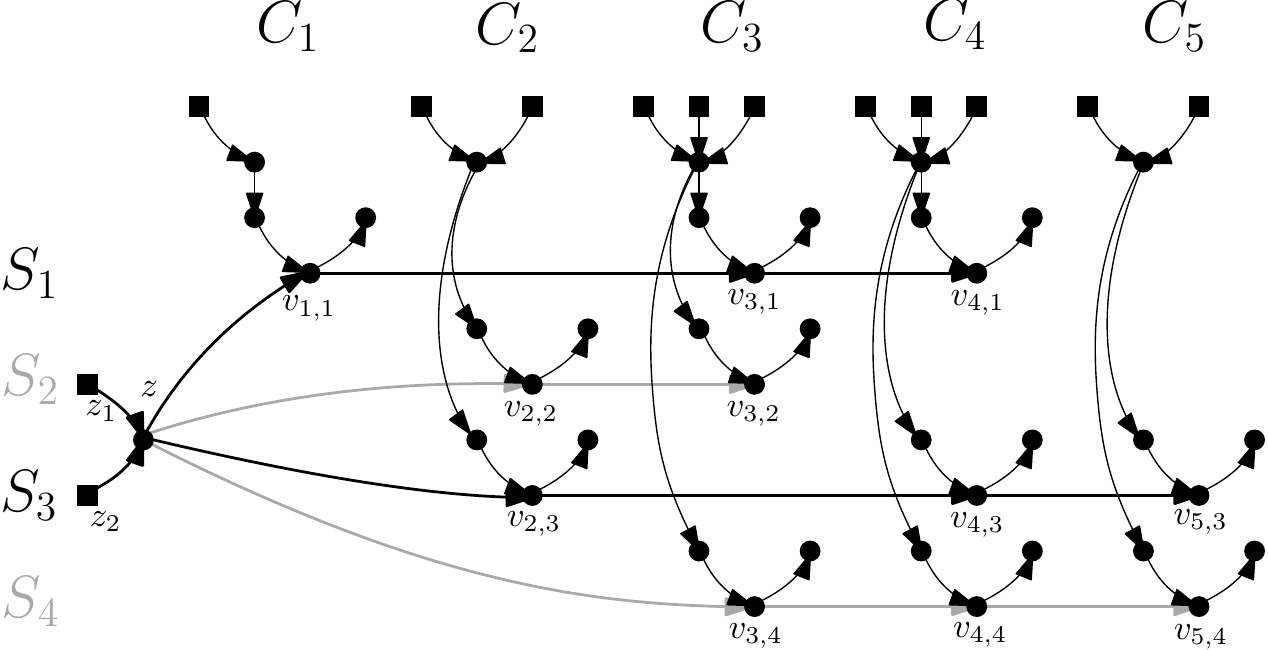}
  \caption{(Cont.\ Figure~\ref{fig:nphard2}) The fact that the set cover contains all elements in $\univ$ guarantees that the corresponding subgraph $H$ of $D$ is connected. For $k=2$, there exists a set cover $\setcovi = \{1, 3\}$ of size $2$, corresponding to sets $S_1$ and $S_3$, respectively.}
\label{fig:nphard3}
\end{center}
\end{figure}

Clearly, by the construction, the set $\cW$ consists of $\mathbf{k}_B=|\source{D_\textrm{SC}}|$ paths that together cover all vertices of $D_{\textup{SC}}$ and each of which starts at a distinct vertex in $\cB=B^{-1}(\natplus)$. Thus, it remains to prove that these paths induce a subgraph $H$ of $D_{\textup{SC}}$ whose underlying graph is connected. First, by the definition of the paths $\path{i}{j}$, observe that each element component induces a subgraph whose underlying graph is connected. Hence, it is enough to argue that for each element component $\comp{i}$, where $i \in \cU$, there exists a directed path in $H$ that connects $z$ with a vertex of $\comp{i}$. Suppose for a contradiction that this is not the case for the $i$-th element component $\comp{i}$, for some $i \in \univ$. That means that no vertex $\verv{i}{j}$, for any $j \in \seti{i}$, is lying on any of the chosen horizontal paths in $H$. But that means, by the choice of the horizontal paths and the fact that $\setcovi$ is a solution to the \problemSC\ problem, that $i$ does not belong to $\bigcup_{j \in \setcovi} S_j$, a contradiction. 
\end{proof}

\subsection{Converse implication} %--------------------------------------------------------------------------------

In this subsection, to complete the NP-completeness proof of the \problemCMR\ problem, we are going to prove the converse implication, in a sequence of lemmas. Recall that in our digraph $D_{\textup{SC}}=(\cV,\cA,F,B)$, we have 
$$\source{D_{\textup{SC}}}=\{z_1,\ldots, z_k\} \cup \bigcup\limits_{i \in \univ} \{u_{i,j} \st j \in I_i\},$$
and we set $B(v)=1$ for each $v\in\source{D_{\textup{SC}}}$ and $B(v)=0$ otherwise, and so $\mathbf{k}_B=|\source{D_{\textup{SC}}}|$. 
Thus, keeping in mind that $D_{\textup{SC}}$ is acyclic, any solution $\cW$ to the \problemCMR\ problem for $D_{\textup{SC}}$ consists of $\mathbf{k}_B$ paths that start at distinct source vertices in $\source{D_{\textup{SC}}}$; in the following, $\pi(v) \in \cW$ denotes the unique path that starts at a source vertex $v\in \source{D_{\textup{SC}}}$. 

\begin{lemma}\label{lem:opt_ind_ec}
Suppose that the \problemCMR\ problem admits a positive answer for the input graph $D_\textrm{SC}$. Then, there exists a solution $\{\pi(v)\st v\in\source{D_\textrm{SC}}\}$ to the \problemCMR\ problem for $D_\textrm{SC}$ such that for each $i\in\univ$ and $j\in\seti{i}$, $V(\pi(\veru{i}{j}))\subseteq V(\comp{i})$.
\end{lemma}
\begin{proof}
Let $\cW$ be a solution to the \problemCMR\ problem for $D_\textrm{SC}$ and assume that in $\cW$, for some $i\in \univ$ and $j \in \seti{i}$, we have $\pi(\veru{i}{j})=(\veru{i}{j},\veruc{i},\verup{i}{j'},\verv{i}{j'},\verv{i'}{j'})\circ P$ for some (possibly empty) path $P$, that is, the arc $(\verv{i}{j'},\verv{i'}{j'})\in A(\pi(\veru{i}{j}))$ and so $\pi(\veru{i}{j})$  `leaves' $C_i$ at $\verv{i}{j}$ by visiting vertex $\verv{i'}{j'}$, for some $i'>i$. We will argue that we may obtain another solution to the \problemCMR\ problem in which $V(\pi(\veru{i}{j}))\subseteq V(\comp{i})$ as required by the lemma. The idea is to modify two paths in $\cW$ maintaining the following invariant: each path that is a subgraph of an element component remains a subgraph of this element component. Thus, since this modification can be repeated for any $i \in \univ$ and $j \in \seti{i}$, this is sufficient to prove our claim.

Since $\cF=F^{-1}(1)=\cV$, we have $\vervp{i}{j'}\in\cF$ and hence there exists another path $\pi(v) \in \cW $ for some $v\in\source{D_\textup{SC}}$ such that $\pi(v)=P' \circ (\verv{i}{j'},\vervp{i}{j'})$, for some path $P'$ in $D_\textup{SC}$. % in the solution to $\problemCMR$.
We then modify the set $\cW$ of paths by removing the two paths $\pi(\veru{i}{j})$ and $\pi(v)$, and then adding the following two paths: $(\veru{i}{j},\veruc{i},\verup{i}{j'},\verv{i}{j'},\vervp{i}{j'})$ and $P' \circ (\verv{i}{j'},\verv{i'}{j}) \circ P$. 

Observe that any two paths that start at vertices in $\setu{i}\setminus\{u_i\}$ have only the vertex $\veruc{i}$ in common --- this is due to the fact that the vertices in $\setup{i}$ are only reachable with paths that start at vertices in $\setu{i}\setminus\{u_i\}$ and $\card{\setup{i}}=\card{\setu{i}\setminus\{u_i\}}=\card{B^{-1}(1)\cap V(\comp{i})}$. Consequently, the vertex set $V(\pi(v))$ of the original path $\pi(v)$ is not a subset of a single element component, and hence, after the modification, each path whose vertex set is a subset of an element component keeps this property as required. Clearly, all vertices of $D$ are covered in the new solution and, since the two new paths also share the vertex $\verv{i}{j'}$, the modified set of paths also induces a connected spanning subgraph of the underlying graph of $D$.
\end{proof}

\begin{lemma}\label{lem:opt_ec}
Suppose that the \problemCMR\ problem admits a positive answer for the input graph $D_\textrm{SC}$. Then, there exists a solution $\{\pi(v)\st v\in\source{D_\textup{SC}}\}$ to the \problemCMR\ problem for $D_\textrm{SC}$ such that:
\begin{itemize}
\item[$a)$] for each $i \in \cU$ and $j\in  \seti{i}$, we have $\pi(\veru{i}{j})=\path{i}{j}$;
\item[$b)$] for each $t\in [k]$, we have $\pi(z_t)=(z_{t},z) \circ \pathh{l}$ for some $l\in [m]$.
\end{itemize}
\end{lemma}
Before we proceed with the proof of Lemma~\ref{lem:opt_ec}, note that Property (a) implies that the paths $\pi(\veru{i}{j})$, where $i\in\univ$ and $j\in\seti{i}$, visit all and only  vertices of all components $\comp{i}$, $i\in\univ$. Therefore, since these components have no vertices in common, the underlying graph becomes connected only thanks to the paths $\pi(v)$, $v\in\setz\setminus\{z\}$, which Property (b) refers to.
\begin{proof}
(a) Consider any solution, say $\cW$, to the \problemCMR\ problem for $D_{\textup{SC}}$ and assume that in $\cW$, for some $i \in \univ$ and $j\in\seti{i}$, we have $\pi(\veru{i}{j})\neq\path{i}{j}$; we shall refer to $\pi(\veru{i}{j})$ as well as to any other path $\pi(\veru{i'}{j'})$ such that $\pi(\veru{i'}{j'}) \neq \path{i'}{j'}$ as {\em inconsistent}. By Lemma~\ref{lem:opt_ind_ec}, $\pi(\veru{i}{j})\subseteq V(\comp{i})$ and hence there exists $j'\in\seti{i}\setminus\{j\}$ such that $\pi(\veru{i}{j})=(\veru{i}{j},\veruc{i},\verup{i}{j'},\verv{i}{j'},\vervp{i}{j'})$.
Then, again by Lemma~\ref{lem:opt_ind_ec} and the fact that each vertex in $\setup{i}$ (in particular $\verup{i}{j}$) has to belong to some path $\pi(v) \in \cW$, $v\in\setu{i}\setminus\{u\}$, there exists in $\cW$ another inconsistent path $\pi(\veru{i}{j''})=(\veru{i}{j''},\veruc{i},\verup{i}{j},\verv{i}{j},\vervp{i}{j})$ for some $j''\in\seti{i}\setminus\{j\}$. Now, we modify the solution by substituting $\pi(\veru{i}{j}):=\path{i}{j}$ and $\pi(\veru{i}{j''}):=(\veru{i}{j''},\veruc{i},\verup{i}{j'},\verv{i}{j'},\vervp{i}{j'})$. Clearly, the two new paths still share vertex $u_i$ and cover exactly the same vertices as the two original ones. Therefore, the modified set of paths is also a solution to the \problemCMR\ problem, moreover, with less number of inconsistent paths.  

By repeating the above replacement argument a finite number of times, if ever needed, we obtain the desired solution satisfying Property (a).
\\[2mm]
(b) Consider any solution, say $\cW$, to the \problemCMR\ problem for $D_{\textup{SC}}$ satisfying already proved Property~(a). Consider any $t \in [k]$. If $\pi(z_t) \in \cW$ ends at a vertex of some horizontal path $\pathh{l}$ for some $l \in [m]$ and $\pi(z_t) \neq (z_t,z) \circ \pathh{l}$, then we just extend $\pi(z_t)$ to have $\pi(z_t)=(z_t,z) \circ \pathh{l}$. Otherwise, by the construction of $D_{\textup{SC}}$, we must have $\pi(z_t)=(z_t, z) \circ P\circ (\verv{i}{l},\vervp{i}{l})$ for some $i\in\univ$, where $P$ is a subpath of $\pathh{l}$ for some $l \in [m]$. By the choice of $\cW$, the arc $(\verv{i}{l},\vervp{i}{l})$ belongs to the (consistent) path $\pi(\veru{i}{j})=\path{i}{j}$ for some $j \in I_i$, and hence the path $\pi(z_t)$ can be replaced by: $\pi(z_t):=(z_t,z) \circ \pathh{l}$. Since $P$ is a subpath of the new path $\pi(z_t)$, we conclude that the new set of paths is also a solution to the \problemCMR\ problem for $D_{\textup{SC}}$, and moreover, it maintains the property that $\pi(\veru{i}{j})=\path{i}{j}$ for each $i \in \univ$ and $j \in I_i$. 

Therefore, by repeating the above replacement argument a finite number of times, if ever needed, we obtain the desired solution satisfying both Properties (a) and (b) 
\end{proof}    
 
Now, we are going to prove our final lemma.
\begin{lemma}\label{lem:from_act_to_st}
If there is a solution to the \problemCMR\ problem for $D_{\textup{SC}}=(\cV,\cA,F,B)$, then there exists a set cover of size $k$ for the set system $(\univ,\allsets)$. 
\end{lemma}
\begin{proof}
By Lemma~\ref{lem:opt_ec}, any solution to the \problemCMR\ problem can be modified to be composed of the following paths: $\pi(\veru{i}{j})=\path{i}{j}$ for each $i\in \univ$ and $j\in\seti{i}$, and $\pi(z_t)=(z_t,z) \circ \pathh{\xi(t)}$ for each $t\in[k]$, where $\xi(t)\in [m]$. Now, we claim that the set $\setcovi = \{\xi(1),\ldots,\xi(k)\}$ is a set cover solution for the instance $\langle\univ,\allsets,k\rangle$.
Indeed, since our solution to the \problemCMR\ problem is valid, for each $i\in\univ$ there exists $t\in [k]$ such that $\verv{i}{j}$ is a vertex of $\pi(z_t)$, since otherwise, in the underlying simple graph induced by our solution, no vertex in $\comp{i}$ is connected by a path to the vertex $z$. Thus, $i\in S_{\xi(t)}$ which completes the proof.
\end{proof}

Note that the \problemCMR\ problem is clearly in NP and, as already observed, the construction of $D_{\textup{SC}}$ is polynomial in the input size to the \problemSC\ problem. Hence, by combining Lemmas~\ref{lem:from_sc_to_act} and~\ref{lem:from_act_to_st}, we obtain the following result.

\begin{theorem} \label{thm:CMRhard}
The \problemCMR\ problem is strongly NP-complete even for directed acyclic graphs $D=(\cV,\cA,F,B)$ with $F^{-1}(1)=\cV$ and $B(v)=1$ if $v$ is a source vertex in $D$ and $B(v)=0$ otherwise.
\qed
\end{theorem}

\paragraph{No pre-specified positions of snow ploughs.} 
We claim that the Snow Team problem with Unspecified snow team bases is also NP-complete. The reduction is exactly the same as for the \problemCMR\ problem. All we need is to observe that if facilities are located at all vertices of the input digraph, then the number of snow ploughs sufficient to solve the \problemCMRk\ problem is bounded from below by the number of source vertices in the digraph, since there must be at least one snow plough at each of its source vertices. Furthermore, without loss of generality we may assume that in any feasible solution of $k$ walks, all snow ploughs are initially located at source vertices. Since in the digraph $D_{\textup{SC}}$ constructed for the proof of Theorem~\ref{thm:CMRhard}, we have $B(v)=1$ if $v$ is a source vertex in $D_{\textup{SC}}$ and $B(v)=0$ otherwise,  we may conclude with the following corollary. 
 
\begin{corollary} \label{cor:CMRUhard}
The \problemCMRk\ problem is strongly NP-complete even for directed acyclic graphs $D=(\cV,\cA,F)$ with $F^{-1}(1)=\cV$ and $k$ being equal the number of source vertices in $D$.
\qed
\end{corollary}

Since by setting $F(v)=1$ for each vertex $v$ of the input digraph, the \problemCMRk\ problem becomes just the Agent Clearing Tree problem (\problemACT) studied in~\cite{ClearingConnectionsByFewAgents}, 
we immediately obtain the following corollary resolving the open problem of the complexity status of \problemACT\ posed in~\cite{ClearingConnectionsByFewAgents}.

\begin{corollary} \label{cor:ACThard}
The \problemACT\ problem is NP-complete.
\qed
\end{corollary}

\section{Open problem} \label{sec:conclusions}
In all of our variants of the Snow Team problem, we assumed that a snow plough can traverse arbitrary number of arcs. However, from a practical point of view, it is more natural to assume that each snow plough, called an {\em $s$-plough}, can traverse and clear only the fixed number $s$ of arcs~\cite{Ntafos1984}. Observe that in this case, the key Lemma~\ref{lem:reduction-transitive} does not hold, which immediately makes our algebraic approach unfeasible for the Snow Team problem with $s$-ploughs, so this variant requires further studies.

\section*{Acknowledgement}
This research has been partially supported by National Science Centre (Poland) grant number 2015/17/B/ST6/01887.

\bibliographystyle{plain}

\bibliography{bibliography}

\end{document}